\documentclass[a4paper,UKenglish,cleveref, autoref, thm-restate]{lipics-v2021}

\hideLIPIcs  

\usepackage{algorithm}
\usepackage{algpseudocode}
\newcommand{\MPal}{\mathsf{MPal}}
\newcommand{\LPal}{\mathsf{LPal}}

\newcommand{\enc}{\mathsf{enc}}
\newcommand{\RMQ}{\mathsf{RMQ}}
\newcommand{\calC}{\mathcal{C}}
\newcommand{\calD}{\mathcal{D}}

\title{Almost succinct representation of maximal palindromes}

\author{Takuya Mieno}{University of Electro-Communications, Japan}{}{}{}
\author{Tomohiro I}{Kyushu Institute of Technology, Japan}{}{}{}
\author{Valerio Stancanelli}{Università di Pisa, Italy}{}{}{}

\authorrunning{T. Mieno et al.} 

\Copyright{Takuya Mieno, Tomohiro I, Valerio Stancanelli} 



\category{} 

\relatedversion{} 

\nolinenumbers 

\EventEditors{John Q. Open and Joan R. Access}
\EventNoEds{2}
\EventLongTitle{42nd Conference on Very Important Topics (CVIT 2016)}
\EventShortTitle{CVIT 2016}
\EventAcronym{CVIT}
\EventYear{2016}
\EventDate{December 24--27, 2016}
\EventLocation{Little Whinging, United Kingdom}
\EventLogo{}
\SeriesVolume{42}
\ArticleNo{23}

\begin{document}

\maketitle

\begin{abstract}
  Palindromes are strings that read the same forward and backward.
  The computation of palindromic structures within strings is a fundamental problem in string algorithms,
  being motivated by potential applications in formal language theory and bioinformatics.
  Although the number of palindromic factors in a string of length $n$ can be quadratic,
  they can be implicitly represented in $O(n \log n)$ bits of space
  by storing the lengths of all maximal palindromes in an integer array,
  which can be computed in $O(n)$ time~[Manacher, 1975].
  In this paper, for any positive constant $\epsilon < 1$,
  we propose a novel $(3(1+\epsilon)n + o(n))$-bit representation of all maximal palindromes in a string,
  which enables $O(1)$-time retrieval of the length of the maximal palindrome centered at any given position.
  The data structure can be constructed in $O(n)$ time and $O(n)$-bit working space from the input string of length $n$.
  Since Manacher's algorithm and the notion of maximal palindromes are widely utilized
  for solving numerous problems involving palindromic structures,
  our compact representation will accelerate the development of more space-efficient solutions to such problems.
  Indeed, as the first application of our compact representation of maximal palindromes,
  we present a data structure of size $O(n)$ bits that can compute 
  the longest palindrome appearing in any given factor of a string of length $n$ in $O(\log n)$ time.
\end{abstract}

\section{Introduction}\label{sec:intro}
A \emph{palindrome} is a string that reads the same forward and backward.
For centuries, palindromes have been enjoyed as elements of wordplay and puzzles.
In recent years, they have also gained attention in Theoretical Computer Science
due to their combinatorial properties and relevance in fields such as formal language theory, computability theory, and bioinformatics
(see also~\cite{GalilS76,KnuthMP77,0086373,KosolobovRS15,Gusfield1997,KolpakovK09,LangermanW16,MBCT2023} and references therein).
From an algorithmic perspective,
one of the most well-known algorithms for detecting palindromic factors is \emph{Manacher's algorithm}~\cite{Manacher75}.
Manacher's algorithm can efficiently enumerate all \emph{maximal palindromes} in a given string in linear time,
where a maximal palindrome is defined as a palindromic factor that cannot be further extended while preserving its center position.
Alternative approaches to enumerate maximal palindromes based on suffix trees~\cite{Weiner73}
are also known and have been applied depending on the context~\cite{Gusfield1997,MienoFI22}.
Manacher's algorithm is particularly notable
for its algorithmic elegance
as well as
for its applicability to general unordered alphabets.
Thanks to these advantages, it plays a central role in many algorithms
involving palindromic structures, even today~\cite{NarisadaHYS20,RubinchikS20,Charalampopoulos22,EllertGS25}.
Furthermore, over the past decade,
several developments have been made in designing data structures
that compactly represent palindromic information and
efficiently support various queries~\cite{RubinchikS18,RubinchikS17,FunakoshiNIBT21,AmirCPR20,MitaniMSH23,MienoF24,Itzhaki26}.

Although various algorithms and data structures related to palindromes have been developed,
it is somewhat surprising that sublinear $o(n)$-word representations of all maximal palindromes remain largely unexplored.
Very recently, Itzhaki~\cite{Itzhaki26} proposed an $O(n)$-bit representation of maximal palindromes in a string of length $n$,
which allows retrieving the length of each maximal palindrome in $O(1)$ time upon a query.
We note that our compact representation has two quantitative advantages over the prior work:
the constant factor of $n$ in the space complexity is
smaller,
and our representation can be computed in linear time and using $O(n)$ bits of working space.

In this paper,
we first provide a simple $3n$-bit encoding for the maximal palindromes in a string of length $n$, which, however, does not support efficient access to their lengths~(Lemma~\ref{lem:encode}).
By using this encoding,
we then propose a data structure of size $3(1+\epsilon)n + o(n)$ bits representing all maximal palindromes in a string
that supports \emph{constant-time} retrieval of the length of each maximal palindrome centered at any given position,
where $\epsilon < 1$ is an arbitrarily small positive constant~(Theorem~\ref{thm:main}).
The data structure can be constructed in $O(n)$ time and using $O(n)$ bits of working space, given a string of length $n$.
We note that, in general,
if a $kn$-bit encoding method to represent maximal palindromes in a string is given,
the size of our data structure becomes $k(1+\epsilon)n + o(n)$ bits.
This implies that
our proposed method can achieve an \emph{almost succinct} representation.
Namely, the space complexity is within a factor of $(1+\epsilon)$ of the information-theoretic lower bound, provided that the underlying encoding is optimal.
In addition, as an application of Theorem~\ref{thm:main},
we present a data structure of size $O(n)$ bits
that can answer any \emph{internal longest palindrome query}~\cite{MitaniMSH23},
which finds the longest palindrome within a query factor,
in $O(\log n)$ time (Theorem~\ref{thm:internal}).

\subsection*{Related work}
Manacher's algorithm~\cite{Manacher75} can output an integer array of linear length, occupying $O(n \log n)$ bits of space,
that stores the lengths of all maximal palindromes in a string $w$ of length $n$.
If we are interested in the variety of \emph{distinct} palindromic factors in $w$, rather than their occurrences,
we can store them using the \emph{palindromic tree}~(a.k.a.~\emph{EERTREE}) data structure,
which requires $O(d\log n)$ bits of space~\cite{RubinchikS18}, where $d$ is the number of distinct palindromic factors in $w$
and is known to be at most $n+1$~\cite{DroubayJP01}.
The palindromic tree of $w$ can be extremely small when $w$ has very few distinct palindromic factors;
however, it does not support locating occurrences of palindromic factors in $w$.
Charalampopoulos et al.~\cite{Charalampopoulos22} presented
an $O(n/\log_{\sigma}n)$-time~(namely, sublinear-time)
algorithm for computing a longest palindromic factor in $w$,
assuming that the alphabet size $\sigma$ is small and that $w$ is given in a \emph{packed representation} occupying $O(n \log \sigma)$ bits of space.
Recently, Mieno and Funakoshi~\cite{MienoF24} proposed an $O(n)$-bit data structure that represents
all \emph{unique} palindromic factors in $w$, where a factor is said to be unique if it occurs exactly once in $w$.

\section{Preliminaries}\label{sec:pre}
\subsection{Intervals, strings, and palindromes}

For non-negative values $i, j$ with $i \le j$,
let $[i.. j] = \{v \in \mathbb{Z} \mid i \le v \le j\}$ be the set of consecutive integers from $i$ to $j$.
Further, let $[i.. j) = [i.. j] \setminus \{j\}$.
Note that $[i.. i) = \emptyset$.
We sometimes refer to a set $[i.. j]$ of consecutive integers as an \emph{interval}.

Let $\Sigma$ be an \emph{alphabet} of size $\sigma$.
An element in $\Sigma$ is called a \emph{character}.
An element in $\Sigma^\star$ is called a \emph{string}.
The length of string $w$ is denoted by $|w|$.
The string of length zero is called the \emph{empty string}.
If $w = xyz$ holds for strings $w, x, y, z \in \Sigma^\star$,
then $x$, $y$, and $z$ are called
a \emph{prefix}, a \emph{factor}, and a \emph{suffix} of $w$,
respectively.
For a string $w$ and integers $i, j \in [0.. |w|)$,
we denote by $w[i]$ the $i$-th character of $w$, and
by $w[i.. j]$ the factor of $w$ that starts at position $i$ and ends at position $j$.
Similar to the notation of intervals, we denote $w[i.. j) = w[i.. j-1]$.
For convenience, let $w[i.. j]$ be the empty string if $i > j$.
Further, let $w[i.. j] = w[i.. |w|-1]$ if $j \ge |w|$.
We denote by $w^R$ the \emph{reversal} of string $w$,
i.e., $w^R = w[|w|-1]w[|w|-2]\cdots w[0]$.
A positive integer $p$ is called a \emph{period} of a string $w$
if $w[i] = w[i+p]$ holds for all $i \in [0.. |w|-p)$.
We also say that $w$ has a period $p$ if $p$ is a period of $w$.
If string $w$ has a period $p^\star$ with $p^\star \le |w|/2$, then
$w$ is said to be \emph{periodic}.

A string $w$ is called a \emph{palindrome} if $w = w^R$ holds.
Note that the empty string is a palindrome.
A factor of a string is said to be \emph{palindromic}
if it is a palindrome.
An occurrence of a palindromic factor $w[i.. j]$ of $w$ is called a \emph{maximal palindrome} 
if (1) $i = 0$, (2) $j = |w|-1$, or (3) $w[i-1] \ne w[j+1]$.
Let $\calC_{w} = [0.. |w|) \cup  \{c+0.5\mid c \in [0.. |w|-1)\}$ be a set of rational numbers.
If $w[i.. j]$ is a palindrome, then we say that the rational number $c = (i+j)/2 \in \calC_w$ is the \emph{center} of $w[i.. j]$,
or $w[i.. j]$ is \emph{centered at} $c$.
For notational simplicity, we sometimes say that a factor is the \emph{MPal at $c$}
if the factor is the maximal palindrome centered at $c$.

\subsection{Model of computation}
In what follows, we arbitrarily fix the input string $w$ of length $n$.
In this paper, when we omit the base of $\log$, it is two.
Our computation model is a standard \emph{word RAM model} of word size $\Omega(\log n)$.
We assume that every character from $\Sigma$ fits within a single machine word,
i.e., given two characters, the equality of them can be determined in constant time.

\subsection{Arrays \texorpdfstring{$\MPal$}{MPal} and \texorpdfstring{$\LPal$}{LPal}}
Manacher~\cite{Manacher75} proposed a linear-time algorithm to compute the longest palindromic factors in a given string,
which inherently computes all the maximal palindromes in the string.
By running Manacher's algorithm, we can obtain a representation of the maximal palindromes in $w$ in linear time.
In this paper, we define the \emph{maximal palindrome array} $\MPal_w$ of length $2n-1$ as follows:
$\MPal_w[2c]$ stores the length of the MPal at $c$ in $w$
for every (half) integer $c\in \calC_w$.
Next, we define the \emph{longest palindrome array} $\LPal_w$ of length $n$ such that
$\LPal_w[j]$ stores the length of the longest palindrome that is a suffix of $w[0.. j]$
for every $j\in [0.. n)$.
We may omit subscripts if they are clear from the context.
The following lemma is known:
\begin{lemma}[Lemma~2 of~\cite{IIT13}]\label{lem:lpal}
  For any strings $x$ and $y$, $\MPal_x = \MPal_y$ iff $\LPal_x = \LPal_y$.
\end{lemma}
Using Lemma~\ref{lem:lpal}, we can compactly encode array $\MPal$.
\begin{lemma}\label{lem:encode}
  Array $\MPal_x$ for a string $x$ can be encoded using $3|x|-2$ bits.
\end{lemma}
\begin{proof}
  By Lemma~\ref{lem:lpal}, it is enough to encode $\LPal_x$ in $3|x|-2$ bits
  instead of encoding $\MPal_x$.
  Consider two positions $j_1, j_2 \in [0.. |x|)$ with $j_1 < j_2$.
  Let $c_1$ and $c_2$ be the center positions of the longest palindromic suffixes of $x[0.. j_1]$ and $x[0.. j_2]$, respectively.
  If the longest palindromic suffix of $x[0.. j]$ is the empty string, we define its center as $j$.
  For the sake of contradiction, assume $c_1 > c_2$.
  On the one hand, $x[2c_1-j_1-1.. j_1]$ is the longest palindromic suffix of $x[0.. j_1]$ by definition.
  On the other hand, $x[2c_2-j_1-1.. j_1]$ is also a palindromic suffix of $x[0.. j_1]$ centered at $c_2$
  since there is a longer palindromic factor $x[2c_2-j_2-1.. j_2]$ centered at $c_2$.
  Then, it holds that $|x[2c_1-j_1-1.. j_1]| = 2(j_1-c_1)+2 < 2(j_1-c_2)+2 = |x[2c_2-j_1-1.. j_1]|$, a contradiction.
  Thus, the center positions of the longest palindromic suffixes of prefixes of $x$ form a non-decreasing sequence of length $|x|$ from $\calC_x$.
  Such a sequence can be encoded using at most $|\calC_x| + (|x|-1) = 3|x|-2$ bits by representing the differences between adjacent values in unary.
  Therefore, $\LPal_x$ can be encoded using $3|x|-2$ bits, and the same holds for $\MPal_x$.
\end{proof}
We provide examples of the encoding in Fig.~\ref{fig:encode}.
Note that the compact encoding in Lemma~\ref{lem:encode} does \emph{not} support constant-time access to each element of $\MPal$.
\begin{figure}[tb]
  \centering
  \includegraphics[width=\linewidth]{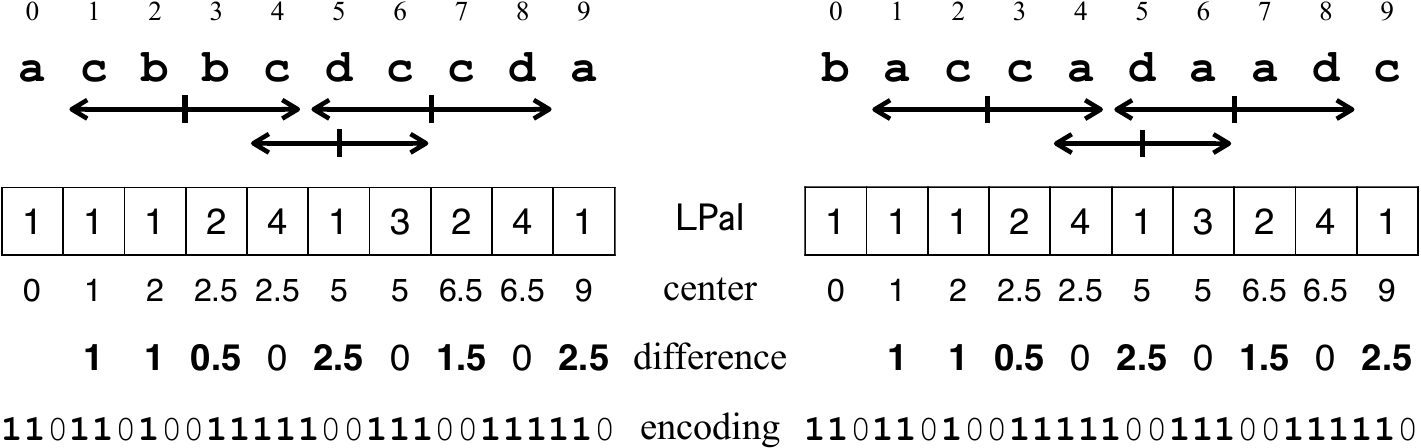}
  \caption{Encodings of two strings $\mathtt{acbbcdccda}$ and $\mathtt{baccadaadc}$.
    Double-headed arrows indicate maximal palindromes that are longer than one.
    For each string, $\LPal$ array, its corresponding centers,
    the differences between adjacent centers, and the binary encoding are shown.
    This encoding represents each difference $d$ in unary:
    specifically, $d$ is encoded as $2d$ ones followed by a single zero.
  }\label{fig:encode}
\end{figure}

\section{\texorpdfstring{$O(n)$}{O(n)}-bit representation of maximal palindromes}
In this section, we show our main theorem:

\begin{restatable}{theorem}{maintheorem}\label{thm:main}
  Given a string $w$ of length $n$ and a positive constant $\epsilon < 1$,
  we can construct in $O(n)$ time and $O(n)$-bit working space
  a data structure of size $3(1+\epsilon)n + o(n)$ bits
  that can return the length of the maximal palindrome centered at $c$ in $O(1)$ time for a given $c \in \calC_w$.
\end{restatable}

In Section~\ref{sec:sketch}, we provide a sketch of our method to introduce the core intuition first.
In the subsequent three subsections, we describe the details of our compact data structure and how to construct it given the $\MPal$ array of the input string.
Finally, in Section~\ref{sec:construct}, we show how to realize the construction using only $O(n)$ bits of working space.

\subsection{Sketch of our method}\label{sec:sketch}

Let $\tau_1$ and $\tau_2$ be integer parameters such that $1 \le \tau_1 < \tau_2$ and $\tau_2 \in O(\log^2 n)$.
We categorize palindromes with respect to their lengths as follows:
a palindrome $P$ is said to be
(1) \emph{short} if $|P| \le 2\tau_1$,
(2) \emph{medium} if $2\tau_1 < |P| \le 2\tau_2$, and
(3) \emph{long} if $|P| > 2\tau_2$, respectively.
Our compact data structure consists of three main components:
$\calD_{\mathrm{long}}$ for long maximal palindromes,
$\calD_{\mathrm{medium}}$ for medium maximal palindromes, and
$\calD_{\mathrm{short}}$ for short maximal palindromes.
Given a center position $c$, we query each component until the length of the MPal at $c$ is obtained.

The idea behind the first two data structures $\calD_{\mathrm{long}}$ and $\calD_{\mathrm{medium}}$ is
the \emph{grouping} of palindromes by using the periodicity of palindromes whose centers are close to each other.
For each parameter $\tau \in \{\tau_1, \tau_2\}$, array $\MPal$ is decomposed into $O(n/\tau)$ blocks of size $\tau$.
Roughly speaking, long/medium palindromes whose centers lie within the same block must be periodic, except for at most one possible exception,
and can be represented compactly using a constant number of arithmetic progressions.
Then, every long/medium maximal palindrome can be restored in constant time, from the arithmetic progressions, upon a query.
The size of $\calD_{\mathrm{long}}$ is $O(\frac{n}{\tau_2}\log n)$ bits
since each block has $O(\log n)$ bits of information.
Furthermore, the size of $\calD_{\mathrm{medium}}$ is $O(\frac{n}{\tau_1}\log \tau_2)$ bits
since a medium palindrome is shorter than $\tau_2$ and each block has $O(\log \tau_2)$ bits of information.
Later, we will choose $\tau_1$ and $\tau_2$ so that $\tau_1 \in \Theta(\log n)$ and $\tau_2 \in \Theta(\log^2 n)$.
Then, the total size of $\calD_{\mathrm{long}}$ and $\calD_{\mathrm{medium}}$ is
$O(\frac{n}{\tau_2}\log n + \frac{n}{\tau_1}\log \tau_2) = O(\frac{n}{\log^2 n}\log n + \frac{n}{\log n}\log(\log^2 n)) \in o(n)$ bits.

The third data structure $\calD_{\mathrm{short}}$ is based on a \emph{lookup table}.
For simplicity, we consider the case where $\epsilon = 0.5$ in this paragraph.
We consider length-$6\tau_1$ factors of $w$ that start at $4k\tau_1$ for $k = 0, 1, 2, \ldots$.
We refer to such factors as \emph{windows}.
By definition, the windows cover the whole string $w$ and any adjacent windows overlap by $2\tau_1$ characters, i.e.,
every short maximal palindrome is a factor of some window.
Given a center position $c \in \calC_w$, we detect a window that contains the (short) maximal palindrome centered at $c$.
We then obtain $\MPal[2c]$ by a single access to a precomputed lookup table $T$.
Since the total number of length-$6\tau_1$ strings over $\Sigma$ is $\sigma^{6\tau_1}$,
a na\"ive implementation of table $T$ leads to an $\Omega(\sigma^{6\tau_1})$-bit data structure,
which may be superlinear unless $\tau_1$ or $\sigma$ is sufficiently small.
However, using a compact representation of $\MPal$ (Lemma~\ref{lem:encode}),
the dependency on $\sigma$ can be eliminated.
Finally, the size of table $T$ can be reduced to
$O(\alpha^{\tau_1}\tau_1^2)$ bits where $\alpha$ is a \emph{constant} that is independent of $\sigma$ and $\tau_1$.
By choosing $\tau_1 = (\log_\alpha n)/\beta \in \Theta(\log n)$ for some $\beta > 1$, the size of $T$ becomes $o(n)$ bits.
Furthermore, by Lemma~\ref{lem:encode}, each window can be encoded using $3(6\tau_1)-2 < 18\tau_1$ bits.
Since there are at most $\lceil \frac{n}{4\tau_1}\rceil$ windows, the total size of the encoding is at most $4.5n + 18\tau_1$ bits.
Therefore, the total size of $\calD_{\mathrm{short}}$ is $4.5n + o(n) = 3(1+\epsilon)n + o(n)$ bits when $\epsilon = 0.5$.

\subsection{Data structure for long palindromes}\label{sec:long}

We separate interval $[0.. n)$ into $O(n/\tau_2)$ sub-intervals
where each sub-interval is of length $\tau_2$.
We refer to such a sub-interval as a \emph{large block}.
Let us focus on $k$-th large block, denoted by $B_k$.
Let $L_k$ be the set of long maximal palindromes
whose centers lie within the $k$-th large block $B_k$.
If $|L_k| \le 2$, then we simply store
their lengths using $O(\log n)$ bits.
Otherwise, we use the periodicity of palindromes whose centers are close to each other,
which is summarized as the following lemma:
\begin{restatable}{lemma}{lemthree}\label{lem:palperiod}
  If $|L_k| \ge 3$, then
  the sequence of lengths of maximal palindromes in $L_k$ sorted by their center positions
  can be represented by at most two arithmetic progressions and at most one integer. 
  Also, the center positions from $L_k$ can be represented by a single arithmetic progression.
\end{restatable}
The essence of this lemma is the same as that of Lemma 12 in~\cite{MatsubaraIISNH09},
which follows from classical results shown by Apostolico, Breslauer, and Galil~\cite{ApostolicoBG95}.
A concrete example of Lemma~\ref{lem:palperiod} is provided in Fig.~\ref{fig:palperiod}.
For completeness, we give a proof of Lemma~\ref{lem:palperiod} in Appendix~\ref{sec:apppendix}.
\begin{figure}[tb]
  \centering
  \includegraphics[width=\linewidth]{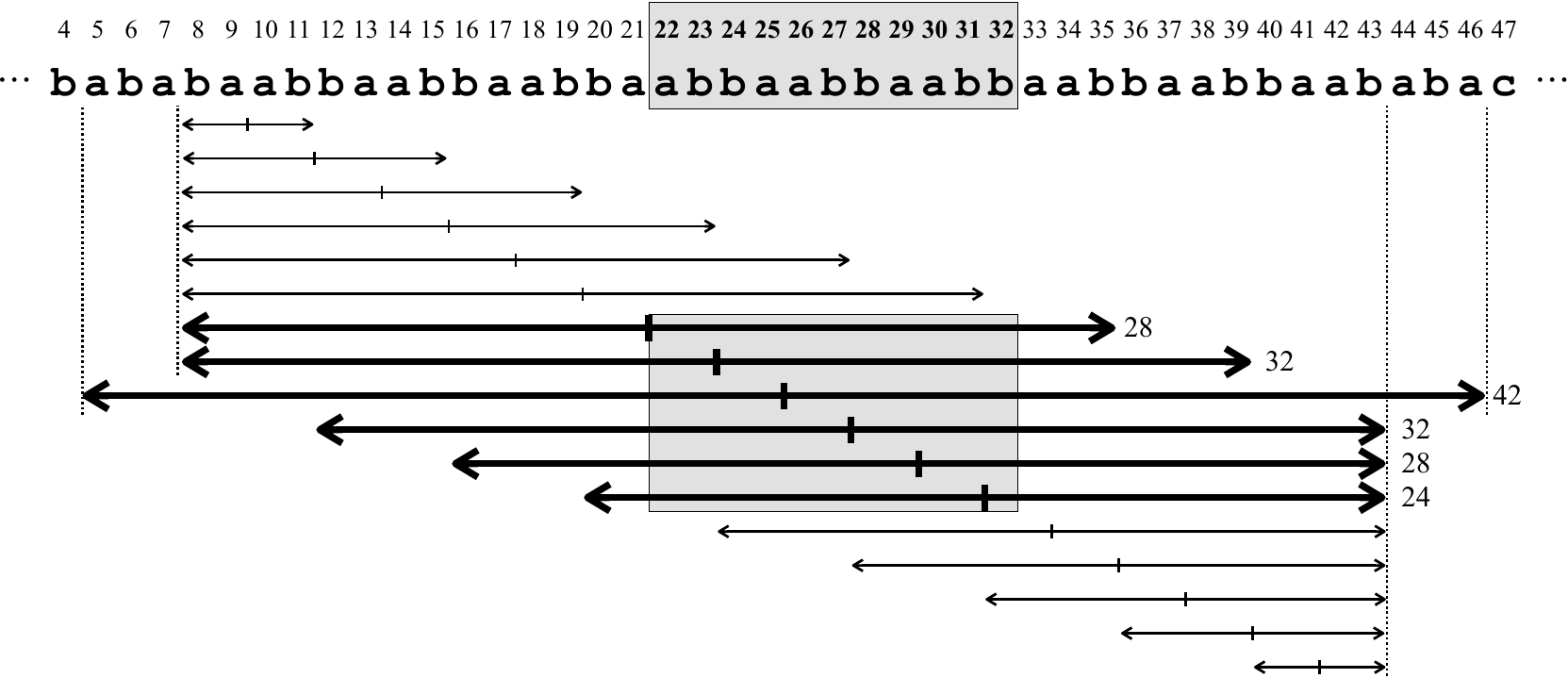}
  \caption{Illustration of a part of an input string and maximal palindromes within it.
    Consider the large block $B = [22.. 32]$ of size $\tau_2 = 11$.
    Maximal palindromes whose centers lie within $B$ are shown with bold arrows,
    and their lengths are labeled to the right.
    The sequence of these lengths is $(28, 32, 42, 32, 28, 24)$,
    which can be described using two arithmetic progressions, $(28, 32)$ and $(32, 28, 24)$,
    along with the single integer $42$.
    Also, the centers of the long palindromes within $B$ are evenly spaced.
  }\label{fig:palperiod}
\end{figure}

By Lemma~\ref{lem:palperiod}, even if $L_k$ is large,
it can be represented in $O(\log n)$ bits and
each element of $L_k$ can be retrieved
by a constant number of arithmetic operations.
Finally, let us consider the construction of the data structure.
Given a string $w$ of length $n$, we first compute the array $\MPal_w$.
For each large block, we scan the corresponding sub-array of $\MPal_w$
and extract the lengths of long maximal palindromes within that sub-array.
By Lemma~\ref{lem:palperiod}, they can be represented by at most three arithmetic progressions.
Clearly, such arithmetic progressions can be computed in $O(\tau_2)$ time for each large block.
Thus, in total, the data structure of Lemma~\ref{lem:long} can be constructed in $O(n)$ time.

From the above discussion, we obtain the following:
\begin{lemma}\label{lem:long}
  There is a data structure of size $O(\frac{n}{\tau_2}\log n)$ bits
  that can return $\MPal[2c]$ in constant time
  if $\MPal[2c] > 2\tau_2$.
  The data structure can be constructed in $O(n)$ time using $O(n \log n)$ bits of working space.
\end{lemma}

\subsection{Data structure for medium palindromes}\label{sec:med}

We separate interval $[0.. n)$ into $O(n/\tau_1)$ sub-intervals
where each sub-interval is of length $\tau_1$.
We refer to such a sub-interval as a \emph{small block}.
The structure of $\calD_{\mathrm{medium}}$ is similar to that of $\calD_{\mathrm{long}}$.
The differences are
(1) the parameter $\tau_2$ in $\calD_{\mathrm{long}}$ is replaced by $\tau_1$, and
(2) each arithmetic progression is represented in $O(\log \tau_2)$ bits, since
the size of a small block is $\tau_1 < \tau_2$ and the length of a medium palindrome is at most $2\tau_2$.

Thus, similar to Lemma~\ref{lem:long}, we obtain the following:
\begin{lemma}\label{lem:med}
  There is a data structure of size $O(\frac{n}{\tau_1}\log \tau_2)$ bits
  that can return $\MPal[2c]$ in constant time
  if $2\tau_1 < \MPal[2c] \le 2\tau_2$.
  The data structure can be constructed in $O(n)$ time using $O(n \log n)$ bits of working space.
\end{lemma}

If we choose $\tau_2 \in \Theta(\log^2 n)$ and $\tau_1 \in \Theta(\log n)$, and
na\"ively store the length of every short maximal palindrome,
which can be represented in $O(\log\tau_1) = O(\log\log n)$ bits,
we obtain a sub-linear $O(n \log\log n/\log n)$-word representation of $\MPal$:
\begin{corollary}\label{cor:sublinear}
  There is a data structure of $O(n \log \log n)$ bits
  that can return $\MPal[2c]$ in $O(1)$ time for a given $c \in \calC_w$.
  The data structure can be constructed in $O(n)$ time using $O(n \log n)$ bits of working space.
\end{corollary}

\subsection{Data structure for short palindromes}\label{sec:short}

In this subsection, we show the following lemma to reduce
the size of the data structure of Corollary~\ref{cor:sublinear} to $O(n)$ bits.
\begin{lemma}\label{lem:short}
  We set $\tau_1 = \lfloor \epsilon \log n/12\rfloor$ for a given positive constant $\epsilon < 1$.
  There is a data structure of size $3(1+\epsilon)n + o(n)$ bits
  that can return $\MPal[2c]$ in constant time
  if $\MPal[2c] \le 2\tau_1$.
  The data structure can be constructed in $O(n)$ time using $O(n \log n)$ bits of working space.
\end{lemma}
\begin{proof}
  Let $\delta = 2/\epsilon > 2$ be a constant.
  For each $k \in [0.. \frac{n}{\delta\tau_1})$, let $X_{k}= w[k\delta\tau_1.. k\delta\tau_1+(2+\delta)\tau_1)$
  be the length-$(2+\delta)\tau_1$ factor of $w$ that starts at position $k\delta\tau_1$.
  We refer to $X_{k}$ as the $k$-th \emph{window}.
  By definition, $X_k$ and $X_{k+1}$ overlap by $2\tau_1$ characters, and thus,
  every short maximal palindrome is a factor of some window.

  By Lemma~\ref{lem:encode}, array $\MPal_x$ of a string $x$ can be encoded using $3|x|-2$ bits.
  We denote by $\enc(X)$ the $(3|X|-2)$-bit encoding of window $X$.
  Let $T$ be a lookup table (empty initially).
  For every window $X$ in $w$, we add records $\langle (\enc(X), p), \ell\rangle$
  to the lookup table $T$ if they do not already exist, where $p \in \calC_X$ is a center position of $X$
  and $\ell$ is the length of the MPal at $p$ in $X$. We only add records corresponding to short palindromes.
  The lookup table requires
  $O(2^{3(2+\delta)\tau_1-2}\cdot(2+\delta)\tau_1\cdot\log((2+\delta)\tau_1)) = O(2^{3(2+\frac{2}{\epsilon})\tau_1}\tau_1 \log \tau_1)$ bits
  since there are at most $2^{3|X|-2} = 2^{3(2+\delta)\tau_1-2}$ possible encodings $\enc(X)$ for strings $X$ of length $(2+\delta)\tau_1$,
  and each value $\ell$ requires $O(\log(2+\delta)\tau_1) = O(\log \tau_1)$ bits.
  If we set $\tau_1 = \lfloor \epsilon\log n/12\rfloor$, then 
  the size of the table is $O(2^{\frac{(1+\epsilon)\log n}{2}}\log n \log\log n) = O(n^{\frac{1+\epsilon}{2}}\log n \log \log n)$ bits,
  which is $o(n)$ bits since $\epsilon < 1$.
  Further, each window can be represented in $3(2+\delta)\tau_1 \le \frac{1+\epsilon}{2}\log n$ bits, fitting in a constant number of machine words.
  Since there are $\lceil \frac{n}{\delta\tau_1} \rceil \sim 6n/\log n$ windows,
  the total space complexity is $3(1+\epsilon)n + o(n)$ bits of space.

  Given a center position $c$ of a short maximal palindrome as a query,
  we detect a window $X$ that must contain that palindrome.
  Then, we look up the record from $T$
  that corresponds to $(\enc(X), c-\mathsf{beg}(X))$
  where $\mathsf{beg}(X)$ is the beginning position of $X$ in $w$.
  Namely, $c-\mathsf{beg}(X)$ denotes the offset of the center position $c$
  from the beginning position of $X$.
  If no such an element exists, the MPal at $c$ is long or medium, so we terminate the procedure.
  Otherwise, we obtain the correct value of $\MPal[2c]$.

  Finally, let us consider the construction of the lookup tables.
  For each window $Z$ of the input string $w$,
  we perform Manacher's algorithm on it.
  Then we can obtain $\MPal_{Z}$ and $\LPal_{Z}$,
  as well as its encoding $\enc(Z)$ in $O(|Z|)$ time~(see also the proof of Lemma~\ref{lem:encode}).
  We then insert the lengths of (short) maximal palindromes from $\MPal_{Z}$ into the lookup table $T$.
  Recall that every record of the form $\langle(\enc(Z), p), L\rangle$ fits within $O(1)$ machine words.
  The total construction time is $O((2+\delta)\tau_1 \times \frac{n}{\delta\tau_1}) = O(n)$.
\end{proof}

Corollary~\ref{cor:main} immediately follows from Lemmas~\ref{lem:long}, \ref{lem:med}, and~\ref{lem:short}.
\begin{corollary}\label{cor:main}
  Given a string $w$ of length $n$ and a positive constant $\epsilon < 1$,
  we can construct in $O(n)$ time and $O(n \log n)$-bit working space
  a data structure of size $3(1+\epsilon)n + o(n)$ bits
  that can return the length of the maximal palindrome centered at $c$ in $O(1)$ time for a given $c \in \calC_w$.
\end{corollary}

Now, the only remaining task to prove Theorem~\ref{thm:main} is to make the working space compact.

\subsection{Construction of the data structure using $O(n)$-bit working space}\label{sec:construct}

In the proof of Lemmas~\ref{lem:long}, \ref{lem:med}, and~\ref{lem:short},
we have shown how to construct the data structure in linear time, but without further bounds on the working space. In particular, at the beginning we compute the array $\MPal$, which requires $O(n \log n)$ bits.
Below, we slightly modify the construction algorithms to complete the proof of Theorem~\ref{thm:main}.

The original version of \emph{Manacher's algorithm} computes the array $\MPal$ from left to right and stores it explicitly: while calculating $\MPal[2c]$, the algorithm might access $\MPal[2c']$ with $c' < c$. For additional details on the algorithm, see~\cite{Manacher75}.

Our goal is to run a modified version of \emph{Manacher's algorithm}, without storing $\MPal$ explicitly. In order for the algorithm to work, when we are calculating $\MPal[2c]$, we need to be able to access $\MPal[2c']$ with $c' < c$ in constant time. This is possible by constructing the data structure used in Theorem~\ref{thm:main} block-wise, while running \emph{Manacher's algorithm}.

Let us consider the data structure for long palindromes. Recall that, in that data structure, we are separating interval $[0.. n)$ into $O(n/\tau_2)$ sub-intervals where each sub-interval is of length $\tau_2$. Let $[l, r]$ be the current sub-interval, that is, the one which contains the center $c$ we are processing. Our idea is to insert the sub-intervals one by one into the data structure: however, when processing center $c$, the calculation of $\MPal[2c']$ with $c'$ in $[l, c)$ is not supported yet by the data structure, because the sub-interval $[l, r]$ has not been inserted yet. Thus, we maintain an array $W_{\text{long}}$ which contains all $\MPal[2c']$ with $c'$ in $[l, c)$. Then, given $c'$, we can find $\MPal[2c']$ in constant time by checking the position $(c'-l)/2$ of the array $W_{\text{long}}$. After calculating $\MPal[2c]$:

\begin{itemize}
  \item if $c \neq r$, we insert $\MPal[2c]$ into $W_{\text{long}}$;
  \item if $c = r$, we insert the sub-interval $[l, r]$ into the data structure, and we clear $W_{\text{long}}$.
\end{itemize}

We handle medium palindromes similarly: only the endpoints of the sub-intervals change, and we can update an array $W_{\text{medium}}$ in the same way as $W_{\text{long}}$. About short palindromes, when we finish processing a window, we can insert it into the data structure; the centers $c$ such that the window contains the longest palindrome centered at $c$ can be erased from $W_{\text{short}}$.

Compared to the algorithm in Theorem~\ref{thm:main}, the bottleneck in the additional space required is the array $W_{\text{long}}$, which stores $O(\log^2 n)$ integers using $O(\log n)$ bits each, so we need $o(n)$-bit additional space. Thus, we use $O(n)$-bit working space during the whole algorithm.

We have completed the proof of our main theorem, restated below:
\maintheorem*

Recall that the factor $3$ in the space complexity is due to the encoding shown in Lemma~\ref{lem:encode}.
As mentioned in Section~\ref{sec:intro}, if a more compact encoding method is invented, the size of our data structure will be reduced accordingly.

\section{Application to internal longest palindrome queries}
In this section, we propose an $O(n)$-bit data structure for solving the following problem.
\begin{definition}[Internal longest palindrome query]\label{def:internalquery}
  For a string $w$ of length $n$,
  the internal longest palindrome query is,
  given an interval $[i.. j] \subseteq [0.. n)$ as a query,
  to return the length of a longest palindromic factor appearing in $w[i.. j]$.
\end{definition}

This problem can be solved in $O(1)$ time per query, using $O(n \log n)$ bits of space~\cite{MitaniMSH23}.
In the following, we propose two alternative approaches 
to answer each query in $O(\log n)$ time using only $O(n)$ bits of space.

\subsection{Approach 1: Based on framework of~\cite{MitaniMSH23}}
The first approach is based on the framework of Mitani et al.'s method~\cite{MitaniMSH23}.
They showed that
the length of a longest palindromic factor of $w[i.. j]$ is
the maximum among the following:
\begin{enumerate}
  \item The length $\ell_p$ of the longest palindromic prefix of $w[i.. j]$,
  \item the length $\ell_s$ of the longest palindromic suffix of $w[i.. j]$, and
  \item the length $\ell_m$ of a longest maximal palindrome of $w$ that is contained within $w[i+1.. j-1]$.
\end{enumerate}
Let $c_p$ (resp.~$c_s$) be the center position of the longest palindromic prefix (resp.~suffix) of $w[i.. j]$.
Namely, $c_p = i + \ell_p/2$ and $c_s = j - \ell_s/2 + 1$ hold.
Once we obtain $c_p$ and $c_s$,
the length $\ell_m$ can be computed by using \emph{range maximum queries} (RMQ) on array $\MPal$.
Precisely, $\ell_m = \MPal[\RMQ_{\MPal}(2c_p+1, 2c_s-1)]$ holds,
where $\RMQ_A(s, t)$ is an element of $\operatorname{argmax}\{A[k]\mid k \in[s.. t]\}$
for array $A$ and its indices $s, t$~(the correctness is shown in~\cite{MitaniMSH23}).
It is known that there is an $O(n)$-bit data structure that can answer an RMQ in constant time~\cite{FischerH11}.
Thus, the remaining task is to show how to compute $\ell_s$ efficiently using an $O(n)$-bit data structure,
since the computation of $\ell_p$ can be treated symmetrically.
The following lemma addresses this:
\begin{lemma}
  There is a data structure of size $O(n)$ bits
  that can compute the length $\ell_s$ of the longest palindromic suffix of $w[i.. j]$
  in $O(\log n)$ time for any given range $[i.. j]$.
\end{lemma}
\begin{proof}
  We first present an $O(n \log n)$-bit data structure, and then we make it compact.
  Let $\mathsf{E}$ be the array of length $2n-1$ such that $\mathsf{E}[2c] = c$ if the MPal at $c$ is empty;
  and otherwise, it stores the ending position of the MPal at $c$.
  We search for the center $c_s$ of the longest palindromic suffix of $w[i.. j]$ by using binary search on $\mathsf{E}$.
  Let $m = (i+j)/2$ be the middle position of the range $[i.. j]$.
  Note that $m \le c_s \le j$ holds by the definition of $c_s$.
  The search range $[m, j]$ exhibits the following monotonicity:
  \begin{itemize}
    \item $\mathsf{E}[\RMQ_\mathsf{E}(2m, 2p)] < j$   holds for every $p \in \{m, m+0.5, \ldots, c_s-0.5\}$, and
    \item $\mathsf{E}[\RMQ_\mathsf{E}(2m, 2q)] \ge j$ holds for every $q \in \{c_s, c_s+0.5, \ldots, j\}$.
  \end{itemize}
  Thus, we can determine $c_s$ in $O(\log n)$ time
  by conducting a binary search within the range $[2m.. 2j]$ over $\mathsf{E}$, making use of an RMQ at each step.
  The desired length $\ell_s$ can be immediately obtained from $c_s$ and $j$.

  Finally, we make the data structure compact.
  As mentioned before, an $O(n)$-bit data structure for RMQs on $\mathsf{E}$ can be constructed, allowing $O(1)$-time queries.
  Moreover, each element of $\mathsf{E}$ can be derived from $\MPal$ as $\mathsf{E}[2c] = c + \MPal[2c]/2$.
  Therefore, by Theorem~\ref{thm:main}, array $\mathsf{E}$ can be represented in $O(n)$ bits.
\end{proof}
In summary, we obtain the following:

\begin{restatable}{theorem}{ILPquery}\label{thm:internal}
  There is a data structure of size $O(n)$ bits
  that can answer each internal longest palindrome query
  in $O(\log n)$ time.
\end{restatable}

\subsection{Approach 2: Based on direct binary search}
In the second approach,
we perform a binary search on the answer. In particular, we want to know whether the answer is greater than or equal to an integer $k$. We can use the following lemma.

\begin{lemma}
  There is a data structure of size $O(n)$ bits
  that can determine in constant time whether the answer to the longest palindrome query for the interval $[i.. j]$ is greater than or equal to $k$, for any given range $[i.. j]$ and integer $k$.
\end{lemma}

\begin{proof}
  Given a center $c$, we can determine the length $f(c)$ of the longest palindromic factor appearing in $w[i.. j]$ and centered at $c$, using the array $\MPal$. Specifically, $f(c) = \min(\MPal[2c], 2(c-i)+1, 2(j-c)+1)$. 

  Now we want to determine whether there exists $c$ in $[i, j]$ with $f(c) \geq k$, that is, whether there exists $c$ in $[i, j]$ such that all the following three conditions are true:
  \begin{enumerate}
    \item $\MPal[2c] \geq k$;
    \item $2(c-i)+1 \geq k$;
    \item $2(j-c)+1 \geq k$.
  \end{enumerate}

  This is equivalent to finding $c$ in $[i+(k-1)/2, j-(k-1)/2]$ with $\MPal[2c] \geq k$. 

  This can be checked by using \emph{range maximum queries} (RMQ) on the array $\MPal$.
  Precisely, we check whether $\MPal[\RMQ_{\MPal}(2i+k-1, 2j-k+1)] \geq k$ holds,
  where $\RMQ_A(s, t)$ is an element of $\operatorname{argmax}\{A[k]\mid k \in[s.. t]\}$
  for array $A$ and its indices $s, t$.

  By Theorem~\ref{thm:main}, array $\MPal$ can be represented in $O(n)$ bits. Moreover, it is known that there is an $O(n)$-bit data structure that can answer an RMQ in constant time~\cite{FischerH11}. So we are using $O(n)$ bits of memory in total.
\end{proof}

Thus, we can determine the answer to an internal longest palindrome query in $O(\log n)$ time
by conducting a binary search within the range $[1.. j-i+1]$ over $k$, making use of an RMQ at each step.
This completes an alternative proof of Theorem~\ref{thm:internal}.

\section{Conclusions}\label{sec:conc}
In this paper, we proposed a data structure of size $3(1+\epsilon)n + o(n)$ bits
that can return the length of the maximal palindrome centered at
a given position in constant time
for a string of length $n$.
As an application of this compact representation,
we presented data structures of size $O(n)$ bits
that can answer any internal longest palindrome query
in $O(\log n)$ time.

Our future work includes the following:
\begin{itemize}
  \item
    Can we develop a \emph{truly} succinct data structure?
    Namely, does there exist an $(L_n + o(n))$-bit representation of $\MPal$ that supports constant-time access, where $L_n$ denotes the information-theoretic lower bound for encoding the maximal palindromes in a string of length $n$?
  \item Can we construct our data structures in $O(n/\log_{\sigma} n)$ time 
    when the alphabet size $\sigma$ is small?
    An $O(n/\log_{\sigma} n)$-time algorithm to compute the \emph{longest} maximal palindrome is known~\cite{Charalampopoulos22}.
  \item Can we implement our data structure so that it outperforms the simple array representation of $\MPal$ in practice?
\end{itemize}

\section*{Acknowledgments}
The authors would like to thank Prof. Nadia Pisanti for helpful discussions.
This work was supported in part by JSPS KAKENHI Grant Numbers JP24K20734 (TM) and JP24K02899 (TI).

\appendix
\section{Proof of Lemma~\ref{lem:palperiod}}\label{sec:apppendix}
In this section, we focus only on palindromes that have \emph{even} lengths.
As discussed in Section 3 of~\cite{ApostolicoBG95},
if we are interested in palindromes that have odd lengths,
we can convert the input string $w$ into
$\hat{w} = w[0]w[0]w[1]w[1]\cdots w[n-1]w[n-1]$,
which is obtained by doubling each character of $w$.
Then, for any $i, j$ with $i, j \in [0.. n)$,
$w[i.. j]$ is a (possibly odd-length) palindrome if and only if $\hat{w}[2i.. 2j+1]$ is an even-length palindrome.
Namely, detecting palindromes in a string of length $n$ can be reduced to detecting \emph{even-length} palindromes in a string of length $2n$.
In the rest of this section,
we simply refer to ``even-length palindromes'' as ``palindromes'' unless otherwise specified.
Furthermore, we  redefine the \emph{center} $c$ of a palindromic factor $w[i.. j]$
as the integer position $c = (i+j)/2 + 0.5$,
which is the first position of the right half of the interval $[i.. j]$.

We next state some known results, which will be used in the proof of Lemma~\ref{lem:palperiod}.
We say that the \emph{radius} of a palindrome is $r$ if the length of the palindrome is $2r$.
We denote by $\gcd(p, q)$ the greatest common divisor of integers $p$ and $q$.

\begin{lemma}[Periodicity Lemma~\cite{fine1965uniqueness}]\label{lem:plemma}
  If a string $w$ has two periods $p$ and $q$ with $p+q-\gcd(p, q) \le |w|$,
  then $w$ also has period $\gcd(p, q)$.
\end{lemma}

\begin{lemma}[Lemma 3.3 in~\cite{ApostolicoBG95}]\label{lem:abg33}
  Assume that there are two palindromic factors in string $w$,
  each with radius at least $r$, centered at $c_1$ and $c_2$, where $c_1 < c_2$ and $c_2-c_1 \le r$.
  Then the factor $w[c_1-r.. c_2+r)$ has period $2(c_2-c_1)$.
\end{lemma}

\begin{lemma}[Lemma 3.4 in~\cite{ApostolicoBG95}]\label{lem:abg:34}
  Assume that the radius of the MPal at $\tilde{c}$ in $w$ is at least $r$.
  Let $w[e_L.. e_R]$ be the maximal factor of $w$ that contains $w[\tilde{c}-r.. \tilde{c}+r)$ and has period $2r$.
  Namely,
  \begin{itemize}
    \item $w[i] = w[i + 2r]$ for $i \in [e_L.. e_R-2r]$,
    \item $e_L = 0$ or $w[e_L-1] \ne w[e_L+2r-1]$, and
    \item $e_R = |w|-1$ or $w[e_R+1] \ne w[e_R-2r+1]$.
  \end{itemize}
  Then the radii of the maximal palindromes centered at $c = \tilde{c} + mr$ for $m \in \mathbb{Z}$
  such that $c \in [e_L.. e_R]$
  are given as follows:
  \begin{itemize}
    \item If $c - e_L < e_R - c + 1$, then the radius is $c-e_L$.
    \item If $c - e_L = e_R - c + 1$, then the radius is larger than or equal to $c - e_L$.
    \item If $c - e_L > e_R - c + 1$, then the radius is $e_R-c+1$.
  \end{itemize}
\end{lemma}

\begin{lemma}\label{lem:overlappal}
  If $w[a.. b] = w[c.. d]$ are identical palindromic factors of $w$ with $a \le c \le b \le d$, then $w[a.. d]$ is a palindrome.
\end{lemma}
\begin{proof}
  For any $i \in [0.. (a+d)/2]$,
  it must hold that $w[a+i] = w[d-i]$ as follows:
  since $w[a.. b]$ is a palindrome, 
  $w[a+i] = w[b-i]$ holds;
  since $w[a.. d]$ has period $c-a = d-b$,
  $w[b-i] = w[b-i+(c-a)] = w[b-i+(d-b)]$ holds;
  finally, $w[b-i+(d-b)] = w[d-i]$.
\end{proof}

Now we are ready to prove Lemma~\ref{lem:palperiod}.
\lemthree*
\begin{proof}
  We denote $N = |L_k|$ and $\tau = \tau_2$ in this proof.
  For each $i \in [1.. N]$, let $P_i$ be the length-$2\tau$ palindrome
  that is obtained by truncating the palindrome in $L_k$ whose center is the $i$-th smallest in $L_k$.
  Further, let $c_i$ be the center of $P_i$.
  Let $r = c_2 - c_1$. Since $c_1, c_2 \in B_k$ with $|B_k| = \tau$, it holds that $r < \tau$.
  By Lemma~\ref{lem:abg33}, string $X = w[c_1 - \tau.. c_2 + \tau )$ has period $2r$.
  Similarly, let $r' = c_3-c_2 < \tau$.
  Then string $Y = w[c_2 - \tau.. c_3 + \tau )$ has period $2r'$.
  Thus, $P_2 = w[c_2 - \tau .. c_2 + \tau )$ has periods $2r$ and $2r'$
  since it is a factor of both $X$ and $Y$.
  For the sake of contradiction, assume that $r \ne r'$.
  We only consider the case $r < r'$ since the other case can be treated symmetrically.
  Since $r+r' = c_3-c_1 < |B_k| = \tau$, $2r + 2r' < 2\tau = |P_2|$ holds,
  and thus,
  $P_2$ has period $r^\star = \gcd(2r, 2r')$ by Lemma~\ref{lem:plemma}.
  Note that $r^\star \le r'$ holds since $r^\star$ is a divisor of $2r'$ and $r^\star \le 2r < 2r'$.
  Since $r^\star$ is a divisor of $2r'$,
  string $Y$ also has period $r^\star$.
  The periodicity of $Y$ leads to an occurrence of a length-$2\tau$ palindrome
  that is identical to $P_2$ and starts at $c_2-\tau+r^\star = c_3-\tau + r^\star-(c_3-c_2) = c_3 - \tau + (r^\star - r')$.
  If $r^\star < r'$, then $c_2-\tau+r^\star < c_3-\tau$,
  which contradicts the definition of $P_3$~(see Fig.~\ref{fig:threepals}).
  Otherwise, if $r^\star = r'$, then $P_2 = P_3$ holds.
  Note that $r^\star$ is an even integer since $r^\star$ is the $\gcd$ of two even integers $2r$ and $2r'$.
  By Lemma~\ref{lem:overlappal}, the MPal at $(c_2+c_3)/2 = c_2 + r^\star/2 < c_3$ is a long palindrome,
  which contradicts the definition of $P_3$.
  Therefore, $r = r'$ holds.
  \begin{figure}[tb]
    \centering
    \includegraphics[width=0.6\linewidth]{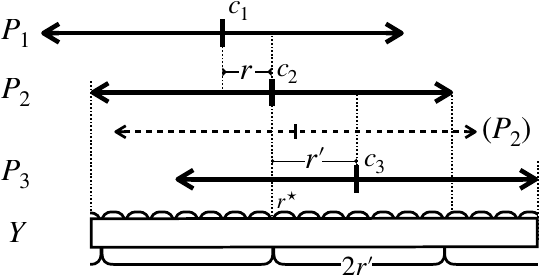}
    \caption{Illustration for a contradiction in the case where $r^\star < r'$.
      String $Y$ has period $r^\star = \gcd(2r, 2r')$.
      If $r^\star$ is strictly smaller than $r'$, another $P_2$, indicated by a dotted arrow,
      occurs between the original $P_2$ and $P_3$, a contradiction.
    }\label{fig:threepals}
  \end{figure}

  In general, the above discussion can be applied to
  any three consecutive length-$2\tau$ palindromes
  $P_{i}$, $P_{i+1}$, and $P_{i+2}$ obtained from $L_k$.
  Thus, by induction, it can be shown that $c_{i+1} - c_i = r$ holds for all $i \in [1.. N)$.
  Namely, $c_i = c_1 + (i-1)r$ holds for all $i \in [1.. N]$.
  Furthermore, since the factor $w[c_i - \tau.. c_{i+1} + \tau )$ has period $2r$ for each $i \in [1.. N)$,
  the factor $w[c_1 - \tau.. c_{N} + \tau )$ also has period $2r$.
  Now, let $w[e_L.. e_R]$ be the maximal factor of $w$ that contains $w[c_1-r.. c_1+r)$ and has period $2r$.
  Then, the inclusion $[c_1 - \tau.. c_{N} + \tau) \subseteq [e_L.. e_R]$ holds,
  and thus, $c_i = c_1 + (i-1)r \in [e_L.. e_R]$ for all $i \in [1.. N]$.
  Therefore, by Lemma~\ref{lem:abg:34}, 
  the radii of maximal palindromes centered at $c_i$ can be represented by
  at most two arithmetic progressions and a single integer.
  More precisely,
  (1) for $c_i$ with $c_i - e_L < e_R - c_i + 1$, the radii can be represented by $c_1 + (i-1)r - e_L$;
  (2) for $c_i$ with $c_i - e_L > e_R - c_i + 1$, the radii can be represented by $e_R - (c_1 + (i-1)r) + 1$; and 
  (3) there is at most one maximal palindrome centered at $c'$ that satisfies $c' - e_L = e_R - c' + 1$.
\end{proof}

\end{document}